\documentclass[a4paper]{scrartcl}

\pdfoutput=1
\usepackage[english]{babel}
\usepackage{authblk}

\usepackage{multirow}
\usepackage{longtable}
\usepackage[utf8]{inputenc}
\usepackage[T1]{fontenc}
\usepackage{lmodern}

\usepackage{amsmath}
\usepackage{amssymb}
\usepackage{units}

\usepackage{mathpartir}
\usepackage{xcolor}

\usepackage{multirow}
\usepackage{alltt}
\usepackage{url}
\usepackage{xspace}
\usepackage{enumerate}

\usepackage{ifthen}

\usepackage{bbm}
\usepackage{stmaryrd}
\usepackage{graphicx}
\usepackage{url}

\usepackage{supertabular}
\usepackage{tabularx}
\usepackage{listings}
\usepackage{subfigure}

\usepackage{floatflt}
\usepackage{tikz}
\usetikzlibrary{petri}
\usetikzlibrary{positioning}
\usetikzlibrary{arrows}
\usepgflibrary{shapes}
\usetikzlibrary{decorations.pathreplacing}
\usetikzlibrary{calc}

\def\+{\kern 0.2ex}

\font\bigttfont = cmtt12 scaled 1120

\def\code#1{\textsf{\small\bfseries #1}}
\def\vars#1{\textsf{\small #1}}

\definecolor{codegreen}{rgb}{0,0.6,0}
\definecolor{codepurple}{rgb}{0.58,0,0.82}

\lstdefinelanguage{Fiacre}{morekeywords={process,is,states,var,init,wait,from,to,select,end,component,port,priority,par,on,loop,if,then,else,unless,property,assert,ltl,||},
  morekeywords={[2]bool,none,sync}, 
morecomment=[s]{/*}{*/}}

\let\s\sigma

\newcommand{\rational}{\mathbb{Q}}
\newcommand{\rationalp}{\rational^{+}}
\newcommand{\nat}{\mathbb{N}}

\newcolumntype{C}{>{${}}c<{{}$}}
\newcolumntype{L}{>{${}}l<{{}$}}
\newcolumntype{R}{>{${}}r<{{}$}}

\def\such{\ .\ }

\let\delay\delta

\tikzstyle{transition} =[very thick, rectangle, draw, inner xsep=2mm, inner ysep=0.75mm]
\tikzstyle{vtransition}=[very thick, rectangle, draw, inner ysep=2mm, inner xsep=0.75mm]

\tikzstyle{place}=[circle, draw, minimum size=4ex]
\tikzstyle{every label}=[font=\sf\footnotesize]

\tikzstyle{pre}+=[>=stealth]
\tikzstyle{post}+=[>=stealth]
\tikzstyle{readarc}=[pre, >=*, shorten <=0pt]
\tikzstyle{prio}=[draw, ->, orange, >=stealth, shorten >=1.25pt, shorten <=1.25pt, densely dashed]
\tikzstyle{poids}=[font=\scriptsize\sf]

\tikzstyle{action}=[rectangle, draw, color=black!80, thin, dotted, inner xsep=0.4ex, inner ysep=0.1ex]

\tikzstyle{gil common}=[color=black, thin, draw=none, fill=none, dash pattern=]

\tikzstyle{gil search}=[gil common, draw, dashed, ->, >=triangle 60]
\tikzstyle{gil strong search}=[gil search, draw, ->>]
\tikzstyle{gil context}=[gil common, draw, {[-)}]

\tikzstyle{gil close context}=[gil common, draw, {[-)}]
\tikzstyle{gil open context}=[gil common, draw, {]-)}]
\tikzstyle{gil dotted context}=[gil common, draw, loosely dotted,-]
\tikzstyle{gil open search}=[gil search, draw, {]->}]
\tikzstyle{gil close search}=[gil search, draw, {[->}]
\tikzstyle{gil time context}=[gil common, draw, |-|]
\tikzstyle{gil open strong search}=[gil search, draw, {]->>}]
\tikzstyle{gil close strong search}=[gil search, draw, {[->>}]

\tikzstyle{gil segment}=[gil common, draw, |-|]
\tikzstyle{gil exists}=[gil common, draw,kite,midway,kite vertex angles=60, inner sep=0.4ex]
\tikzstyle{gil always}=[gil common, draw, gil exists, rectangle, inner ysep=0.8ex]
\tikzstyle{gil boite}=[gil common, draw, dotted, thick, rounded corners=4pt]
\tikzstyle{delta}=[gil common, draw,thick,isosceles triangle, anchor=apex, rotate=90, inner sep=0.4ex]

\tikzstyle{gil interval}=[gil common, draw, decorate,decoration={brace}]

\tikzstyle{every pin}=[gil common, pin distance=0.4ex]
\tikzstyle{every pin edge}=[gil common, draw,-]

\tikzstyle{audessus}=[gil common, anchor=south east, inner sep=0em, yshift=1.2ex]
\tikzstyle{audessous}=[gil common, anchor=north east, inner sep=0em, yshift=-1.2ex]
\tikzstyle{inlaudessous}=[gil common, anchor=east, inner sep=1ex, yshift=-0.5ex]
\tikzstyle{inlaudessous2}=[gil common, anchor=east, inner sep=2.2ex, yshift=-0.5ex]

\tikzstyle{timint}=[gil common, midway,above=-0.5ex]

\tikzstyle{snippet}=[gil common, draw, rectangle, rounded corners=0pt, inner ysep=0.15ex, inner xsep=1ex, semithick,
                     densely dotted, draw=black, text width=]

\tikzstyle{patterncommon}=[anchor=text, draw=black]

\tikzstyle{patterntitle}=[patterncommon,rectangle,rounded corners=2pt, inner ysep=0.75ex, inner xsep=3ex]
\tikzstyle{patternexample}=[patterncommon,rectangle,draw=none, inner ysep=0.25ex, inner xsep=0.8ex]

\newcommand{\sem}[1]{[\kern-.5mm[{#1}]\kern-.5mm]}
\newcommand{\interp}[2][{}]{(\kern-.5mm({#2})\kern-.5mm)_{#1}}

\newcommand{\eqdef}{\ensuremath{\overset{\text{\tiny def}}{=}}}

\newcommand{\tick}{\ensuremath{\mathit{Tick}}}
\newcommand{\kleene}{\ensuremath{\mathclose{\overset{*}{\ }}}}
\newcommand{\ltl}[1]{\texttt{#1}}
\newcommand{\ltland}{\,\ensuremath{\wedge}\,}
\newcommand{\ltlo}{\,\ensuremath{\text{\texttt{o}}}\,}
\newcommand{\ltls}{\,\ensuremath{\text{\texttt{*}}}\,}
\newcommand{\ltlor}{\,\ensuremath{\vee}\,}

\makeatletter
\newcommand{\makeord}[1]{
  \edef\@tempa{\the\mathcode`#1 }
  \begingroup\lccode`~=`#1
  \lowercase{\endgroup\edef~}{\mathpunct{\mathchar\@tempa}}
  \mathcode`#1="8000
}

\usepackage{amsthm}

\newtheorem{lemma}{Lemma}

\makeatother
\makeatletter
\newtoks\fintableau
\let\fintableau\@arraycr
\makeatother

\begin{document}


\title{Automating the Verification of\\ Realtime Observers using
  Probes and\\ the Modal mu-calculus\thanks{This work was presented at TTCS 2015, the First IFIP International
Conference on Topics in Theoretical Computer Science, August 26-28,
2015.  Institute for Research in Fundamental Sciences (IPM), Tehran,
Iran.}}
\author[1,2]{Silvano Dal Zilio}
\author[1,2]{Bernard Berthomieu}
\affil[1]{CNRS, LAAS, F-31400 Toulouse,  France}
\affil[2]{Univ de Toulouse, LAAS, F-31400 Toulouse, France}
\date{}
\maketitle
\begin{abstract}
  A classical method for model-checking timed properties---such as
  those expressed using timed extensions of temporal logic---is to
  rely on the use of observers. In this context, a major problem is to
  prove the correctness of observers. Essentially, this boils down to
  proving that: (1) every trace that contradicts a property can be
  detected by the observer; but also that (2) the observer is
  innocuous, meaning that it cannot interfere with the system under
  observation. In this paper, we describe a method for automatically
  testing the correctness of realtime observers. This method is
  obtained by automating an approach often referred to as \emph{visual
    verification}, in which the correctness of a system is performed
  by inspecting a graphical representation of its state space. Our
  approach has been implemented on the tool Tina, a model-checking
  toolbox for Time Petri Net.
\end{abstract}


\section{Introduction}
\label{sec:introduction}

A classical method for model-checking timed behavioral
properties---such as those expressed using timed extensions of
temporal logic---is to rely on the use of observers. In this approach,
we check that a given property, \vars{P}, is valid for a system
\vars{S} by checking the behavior of the system composed with an
observer for the property. That is, for every property \vars{P} of
interest, we need a pair $(\vars{Obs}_P, \phi_{P})$ of a system (the
observer) and a formula. Then property \vars{P} is valid if and only
if the composition of \vars{S} with $\vars{Obs}_P$, denoted (\vars{S}
\code{||} $\vars{Obs}_P$), satisfies $\phi_{P}$. This approach is
useful when the properties are complex, for instance when they include
realtime constraints or involve arithmetic expressions on
variables. Another advantage is that we can often reduce the initial
verification problem to a much simpler model-checking problem, for
example when $\phi_{P}$ is a simple reachability property.

In this context, a major problem is to prove the correctness of
observers. Essentially, this boils down to proving that every trace
that contradicts a property can be detected. But this also involve
proving that an observer will never block the execution of a valid
trace; we say that it is \emph{innocuous} or non-intrusive. In other
words, we need to assure that the ``measurements'' performed by the
observer can be made without affecting the system.

In the present work, we propose to use a model-checking tool chain in
order to check the correctness of observers. We consider observers
related to linear time properties obtained by extending the pattern
specification language of Dwyer et al.~\cite{ksu} with hard, realtime
constraints. In this paper, we take the example of the pattern
``\code{Present} \vars{a} \code{after} \vars{b} \code{within} $[d_1,
d_2[$'', meaning that event \vars{a} must occur within $d_1$ units of
time (u.t.) of the first occurrence of \vars{b}, if any, but not later
than $d_2$. Our approach can be used to prove both the soundness and
correctness of an observer when we fix the values of the timing
constraints (the values of $d_1$ and $d_2$ in this particular
case). 

Our method is not enough, by itself, to prove the correctness of a
verification tool. Indeed, to be totally trustworthy, this will
require the use of more heavy-duty software verification methods, such
as interactive theorem proving. Nonetheless our method is
complementary to these approaches. In particular it can be used to
debug new or optimized definitions of an observer for a given property
before engaging in a more complex formal proof of its correctness.

Our method is obtained by automating an approach often referred to as
\emph{visual verification}, in which the correctness of a system is
performed by inspecting a graphical representation of its state
space. Instead of visual inspection, we check a set of branching time
(modal $\mu$-calculus) properties on the discrete time state space of
a system. These formulas are derived automatically from a definition
of the pattern expressed as a first-order formula over timed
traces. The gist of this method is that, in a discrete time setting,
first-order formulas over timed traces can be expressed,
interchangeably, as regular expressions, LTL formulas or modal
$\mu$-calculus formulas.

This approach has been implemented on the tool Tina~\cite{tina}, a
model-checking toolbox for Time Petri Net~\cite{merlin} (TPN). This
implementation takes advantage of several components of Tina: state
space exploration algorithms with a discrete time semantics (using the
option \texttt{-F1} of Tina); model-checkers for LTL and for modal
$\mu$-calculus, called \emph{selt} and \emph{muse} respectively; a new
notion of \emph{verification probes} recently added to
Fiacre~\cite{Fiacre07,filfmvte2008}, one of the input specification
language of Tina. While model checkers are used to replace visual
verification, probes are used to ensure innocuousness of the
observers.

\subsection{Outline and contributions} The rest of the paper is
organized as follows. In Sect.~\ref{sec2:fiacre}, we give a brief
definition of Fiacre and the use of probes and observers in this
language. In Sect.~\ref{sec3:timedtrace}, we introduce the technical
notations necessary to define the semantics of patterns and time
traces and focus on an example of timed patterns. Before concluding,
we describe the graphical verification method and show how to use a
model-checker to automatize the verification process\footnote{Code is
  available at
  \url{http://www.laas.fr/fiacre/examples/visualverif.html}}.

The theory and technologies underlying our verification method are not
new: model-checking algorithms, semantics of realtime patterns,
connection between path properties and modal logics, \dots\
Nonetheless, we propose a novel way to combine these techniques in
order to check the implementation of observers and in order to replace
traditional ``visual'' verification methods that are prone to human
errors. 

Our paper also makes some contributions at the technical level. In
particular, this is the first paper that documents the notion of
probe, that was only recently added to Fiacre. We believe that our
(language-level) notion of probes is interesting in its own right and
could be adopted in other specification languages.

\section{The Fiacre Language}
\label{sec2:fiacre}

We consider systems modeled using the specification language
Fiacre~\cite{Fiacre07,filfmvte2008}. (Both the system and the
observers are expressed in the same language.)  Fiacre is a
high-level, formal specification language designed to represent both
the behavioral and timing aspects of reactive systems.

Fiacre programs are stratified in two main notions: \emph{processes},
which are well-suited for modeling structured activities, and
\emph{components}, which describes a system as a composition of
processes. Components can be hierarchically composed. We give in
Fig.~\ref{fig/fiacre-process} a simple example of Fiacre specification
for a computer mouse button capable of emitting a double-click
event. The behavior, in this case, is to emit the event \vars{double}
if there are more than two \vars{click} events in {strictly less} than
one unit of time (u.t.).
\begin{figure}[tbh]
  \centering
  \begin{tabular}{c|@{\quad}c}
    \begin{minipage}[t]{0.42\linewidth}
      {
\begin{lstlisting}
process Push [click  : none,
              single : none, 
              double : none,
              delay  : none] is

 states s0, s1, s2

 var dbl : bool := false

 from s0 click; to s1

 from s1
  select
     click; dbl := true; loop
  [] delay; to s2
  end

 from s2
  if dbl then double
  else single end;
  dbl := false; to s0
\end{lstlisting}}
    \end{minipage}
    &
    \begin{minipage}[t]{0.42\linewidth}
      {
\begin{lstlisting}
component Mouse [click  : none,
                 single : none,
                 double : none] is

   port delay : none in [1,1]

   priority delay > click

   par 
      Push [click, single, double, delay] 
   end
\end{lstlisting}}
    \end{minipage}\\
  \end{tabular}\\
  \caption{A double-click example in Fiacre}
  \label{fig/fiacre-process}
\end{figure}


\subsection{Processes}
A process is defined by a set of parameters and {control states}, each
associated with a set of \emph{complex transitions} (introduced by the
keyword \code{from}). The initial state of a process is the state
corresponding to the first \code{from} declaration.

Complex transitions are expressions that declares how variables are
updated and which transitions may fire. They are built from
deterministic constructs available in classical programming languages
(assignments, conditionals, sequential composition, \dots);
non-deterministic constructs (such as external choice, with the
\code{select} operator); communication on ports; and jump to a state
(with the \code{to} or \code{loop} operators).

For example, in Fig.~\ref{fig/fiacre-process}, we declare a process
named \vars{Push} with four communication ports (\vars{click} to
\vars{delay}) and one local boolean variable, \vars{dbl}.  Ports may
send and receive typed data.  The port type \vars{none} means that no
data is exchanged; these ports simply act as synchronization
events. Regarding complex transitions, the expression related to state
\vars{s1} of \vars{Push}, for instance, declares two possible
transitions from \vars{s1}: (1) on a \vars{click} event, set
\vars{dbl} to true and stay in state \vars{s1}; and (2) on a
\vars{delay} event, change to state \vars{s2}.

\subsection{Components} A {component} is built from the parallel
composition of processes and/or other components, expressed with the
operator \code{par} \vars{P}$_0$ \code{||} \dots \code{||}
\vars{P}$_n$ \code{end}. In a composition, processes can interact both
through synchronization (message-passing) and access to shared
variables (shared memory).

Components are the unit for process instantiation and for declaring
ports and shared variables. The syntax of components allows to
associate timing constraints with communications and to define
priorities between communication events. The ability to express
directly timing constraints in programs is a distinguishing feature of
Fiacre. For example, in the declaration of component \vars{Mouse} (see
Fig.~\ref{fig/fiacre-process}), the \code{port} statement declares a
local event \vars{delay} and asserts that a transition from \vars{s1}
to \vars{s2} should take exactly one unit of time. (Time passes at the
same rate for all the processes.) Additionally, the \code{priority}
statement asserts that a transition on event \vars{click} cannot occur
if a transition on \vars{delay} is also possible.
 
\subsection{Probes and Observers} 
The Fiacre language has been extended, recently, to allow the
definition of {observers}, which are a distinguished category of
sub-programs that interact with other Fiacre components only through
the use of \emph{probes}. A probe is used to observe modifications in
the system without interfering with it; probes react to the occurrence
of an event without engaging in it.

A typical probe declaration is of the form
\vars{path}\code{/}\vars{obs}, where \vars{obs} denotes the observable
and \vars{path} defines its context, that is a path to the component
(or process) instance where \vars{obs} is defined (see for example
{\footnotesize\url{http://www.laas.fr/fiacre/properties.html}}).  In
our setting, observable events are instantaneous actions involved in
the evolution of the system: it can be a synchronization over a port
\vars{p} (denoted \vars{event~p}); a process that enters the state
\vars{s} (denoted \vars{state~s}); or an expression including shared
variables, say \vars{exp}, that changes value (denoted
\vars{value~exp}). For instance, in the case of the \vars{Mouse}
component of Fig.~\ref{fig/fiacre-process}, a probe triggered when the
(only instance of) process \vars{Push} is in state \vars{s2} would
have the form \vars{(Mouse/1/state s2)}. 

The use of probes greatly simplifies the proof of innocuousness of an
observer. In particular, with probes, an observer can only influence a
system by ``blocking the evolution of time'', that is by performing an
infinite sequence of actions in finite time. Therefore, proving that
an observer is innocuous amounts to prove that it has no Zeno
behaviors, which is always possible when a system is bounded.

\begin{figure}[tbh]
  \centering
  \begin{tabular}{l|@{\quad}l}
    \begin{minipage}[t]{0.42\linewidth}
      {
\begin{lstlisting}
process NeverTwice [a:sync] is
 
 states idle, once, error

 from idle a ; to once

 from once a ; to error
\end{lstlisting}}
    \end{minipage}
    &
    \begin{minipage}[t]{0.42\linewidth}
      {
\begin{lstlisting}
component Obs is
 
 port p:sync is Mouse/event click

 par NeverTwice [p] end
    
\end{lstlisting}}
    \end{minipage}
  \end{tabular}
  \caption{A simple observer example\label{fig/fiacre-obs}}
\end{figure}

An observer is a Fiacre component where ports are associated to probes
(using the keyword \code{is}); ports associated with a probe have the
reserved type \vars{sync}. We give a naive example of observer in
Fig.~\ref{fig/fiacre-obs}, where the component \vars{Obs} monitors
synchronizations on the event \vars{click}. In this example, the
process \vars{neverTwice} will reach the state \vars{error} if its
probe parameter, \vars{a}, is triggered more than once.

In the remainder of the text, we use the notation (\vars{Mouse}
\code{||} \vars{Obs}) to denote the program obtained by concatenating
the declaration of these two components (i.e. the code from
Fig.~\ref{fig/fiacre-process} with the code from
Fig.~\ref{fig/fiacre-obs}). As a consequence, we are able to detect if
the system can emit two single click events just by checking if the
process \vars{neverTwice} can reach the state \vars{error} in
(\vars{Mouse} \code{||} \vars{Obs}). This can be easily achieved using
an LTL model-checker (the \emph{selt} tool in our case) with the
property \texttt{[]-(Obs/1/state error)} (meaning that never
\vars{Obs/1} is in state \vars{error}).

\section{Timed Traces and First-Order Formulas over Traces}
\label{sec3:timedtrace}

The semantics of Fiacre (and the properties we want to check) are
based on a notion of \emph{timed traces}, which are sequences mixing
events and time delays. In this context, a ``realtime property'' can
be defined as a set of timed traces, which define timing and
behavioral constraints on the acceptable execution of a system. In
this work, we consider properties derived from realtime patterns, that
can be expressed using first-order formulas over timed traces.

\subsection{Timed Traces} 
In our context, observable events are: communication on a port; the
change of state of a process; and the change of value of a
variable. We use a dense time model, meaning that we consider rational
time delays and work both with strict and non-strict time
bounds. Hence a timed trace is a (possibly infinite) sequence of
events $a, b, \dots$ and durations $\delay \in \rationalp$:
\[\sigma \ ::=\ \epsilon \ \mid\ \sigma \, a \ \mid \ \sigma \,
\delay\]

Given a finite trace $\s$ and a---possibly infinite---trace $\s'$, we
denote $\s \s'$ the \emph{concatenation} of $\s$ and $\s'$.  We will
also use the expression $\Delta(\s)$ to denote the duration (time
length) of a trace $\s$. The semantics of a system expressed with
Fiacre, say \vars{S}, can be defined as a set $\sem{\vars{S}}$ of
timed traces. We use the notation $\s \models \vars{S}$ when the trace
$\s$ is in the set $\sem{\vars{S}}$.  The semantics of a property
(timed pattern) will be expressed as the set of all timed traces where
the pattern holds.  We say that a system \vars{S} satisfies a timed
requirement \vars{P} if $\sem{\vars{S}} \subseteq \sem{\vars{P}}$.

\subsection{Realtime Properties and their Semantics}
We propose to define properties using First-Order Formulas over Timed
Traces (FOTT). A FOTT formula $\Phi(\vec x)$, with free variables
$\vec x = (x_1, \dots, x_n)$, is a first-order logic formula over
traces with equality between traces ($\s = \s'$), comparison between a
duration and an interval ($\Delta(\s) \in I$) and concatenation $(\s =
\s_1 \, \s_2$).
\[
\Phi(\vec x) \ ::=\ \Phi \wedge \Phi' \ \mid\ \neg \Phi \ \mid \
\exists x \such \Phi \ \mid\ (x = \s) \ \mid\ (x = y \, z) \ \mid \
(\Delta(x) \in I) \]
For instance, when referring to a timed trace $\s$ and an event $a$,
the following formula is a tautology if the event $a$ does not occur
in $\s$:
\[
(a \notin \s) \ \eqdef\ \neg \left ( \exists x_1, x_2, x_3 \such (\s =
  x_1 \, x_2) \wedge (x_2 = a\, x_3) \right ) \]
Likewise, we can define the ``scope'' $\s$ \code{after} $b$---that
determines the part of a trace $\s$ located after the first occurrence
of $b$---as the trace $\s'$ denoted by the first-order formula:
$\exists x, y \such (\s = x\, y) \wedge (y = b\, \s') \wedge \left (b
  \notin x \right )$.

The semantics of a formula $\Phi(x_1, \dots, x_n)$ is a set of
valuation functions $\varsigma$ associating a trace $\s_i =
\varsigma(x_i)$ to each of the variable $x_i$ with $i \in 1..n$, also
denoted $[ x_i \mapsto \s_i ]_{i \in 1..n}$. The semantics of $\Phi$
can be defined inductively as follows:
\[
\begin{array}{lcl@{\quad}lcl}
  \sem{\Phi(\vec x) \wedge \Psi(\vec x)}  &=&  \sem{\Phi(\vec x)} \cap
  \sem{\Psi(\vec x)} &
  \sem{x = \s}  &=&  \{ \varsigma
  \mid \varsigma(x) = \s \}
  \\
  \sem{\exists y \such \Phi(\vec x)}  &=& \{ \varsigma
  \mid \varsigma + [y \mapsto \s] \in \sem{\Phi(\vec x)} \} &
  \sem{x = y \, z}  &=&  \{ \varsigma
  \mid \varsigma(x) = \varsigma(y) \, \varsigma(z) \}
  \\
  \sem{\Delta(x) \in I}  &=&  \{ \varsigma
  \mid \Delta(\varsigma(x)) \in I \}
  \\
\end{array}
\]

With these definitions, a \emph{regular set of time traces} is the set
of traces ``solutions'' of an existential FOTT formula with a single
free variable, $\Phi(x)$; that is the set of traces $\s$ such that the
valuation $[x \mapsto \s]$ is in $\sem{\Phi(x)}$.

In this paper, we will mainly restrict ourselves to the special case
of timed traces where events occur at integer dates; i.e. we restrict
delays $\delay$ to be in $\nat$ rather than in $\rationalp$. These
traces can be generated using a ``discrete time'' abstraction of the
models, where special transitions (labeled with \vars{t}) are used to
model the flow of time. Label \vars{t} stands for the ``tick'' of the
logical clock.


The discrete time semantics will be enough to prove all the properties
needed in our study. Indeed, when a model contain only ``closed timing
constraints'' (of the kind $[d_1, d_2]$ or $[d_1, \infty[$), the
discrete time semantics is enough to check reachability
properties. Actually it is enough to check every formulas in the
existential fragment of CTL\kleene\ without next
operator~\cite{janowska2011towards}.

With discrete time, a delay $\delta$ can be replaced by sequences of
$\delta$ \vars{t}'s, and therefore a finite timed trace can be simply
interpreted as a word. In the remainder, we also consider a special
symbol, \vars{z}, that stands for internal actions of the
system. Hence it is possible to interpret the semantics of (discrete)
FOTT specification as a language over the alphabet $A = \{\vars{z},
\vars{t}, \vars{a}, \vars{b}, \dots \}$. Actually, in the discrete
case, we can show that a regular set of time traces is also a regular
language. For example, the semantics of the formula $\exists y, z, w
\such \left ( (x = y \, z) \wedge (z = a \, w) \right)$ is the regular
language corresponding to the expression $A\kleene \cdot \vars{a}
\cdot A\kleene$.

This connection between different type of logics is at the core of our
approach. Our method could be applied to more high-level property
languages, such as timed extension of temporal logic~\cite{SRPMTL},
but would require a more complex encoding into LTL when modalities can
be nested.

\subsection{Our Running Example: the Present Pattern}
Users of Fiacre have access to a catalog of specification patterns
based on a hierarchical classification borrowed from
Dwyer~\cite{ksu}. Patterns are built from five basic
categories---existence, absence, universality, response and
precedence---and can be composed using logical connectives.  In each
category, generic patterns may be specialized using \emph{scope
  modifiers}---such as before, after, between---that limit the range
of the execution trace over which the pattern must hold. Finally,
timed patterns are obtained using one of two possible kind of
\emph{timing modifiers} that limit the possible dates of events
referred in the pattern: \textbf{within $I$}---used to constraint the
delay between two given events to be in the time interval $I$---and
\textbf{lasting $d$}---used to constraint the length of time during
which a given condition holds (without interruption) to be greater
than $d$. 

Due to limited space, we study only one example of timed pattern,
namely {\code{Present} \vars{a} \code{after} \vars{b} \code{within}
  $[d_1, d_2[$}. A complete catalog is available in~\cite{FRP11}.
This is a simple example of {existence} patterns. Existence patterns
are used to express that, in every trace of the system, some events
must occur.  This pattern holds for traces such that the event
\vars{a} occurs at a date $t_0$ after the first occurrence of \vars{b}
with $t_0 \in [d_1, d_2[$. The property is also satisfied if \vars{b}
never holds. Hence traces $\sigma$ that satisfy this pattern are
models of the existential FOTT formula:
\[
\mathrm{Pres}(x) \ \eqdef\ (b \notin x) \ \vee\ \exists y, z, w \such
\left ( (x = y \, b \, z \, a\, w) \wedge (b \notin y) \wedge
  (\Delta(z) \in [d_1, d_2[) \right ) 
\]

\lstinputlisting[float=tbph,label=lproc1,frame=single,captionpos=b,caption=\protect{Observer
  for the pattern: \code{Present} \vars{a} \code{after} \vars{b}
  \code{within} $[d_1,d_2[$}]{presentafter1.txt}

With the discrete semantics, formula $\mathrm{Pres}(x)$ matches
exactly the words of the form $w_1 \, \vars{b} \, w_2 \, \vars{a} \,
w_3$ where $w_1$ contains no occurrences of \vars{b} and $w_2$
contains exactly $k$ occurrences of \vars{t} with $k \in [d_1,
d_2[$. (This is a regular language.) We show in the next section how
to (semi-)automatically generate the regular expression corresponding
to such FOTT formulas.

We give an example of observer associated to this pattern in
Listing~\ref{lproc1}. This observer is composed of one process that
monitors the system through the ports \vars{a} and \vars{b} (that
should be instantiated with the relevant probes). The process is
initially in state \vars{idle} and moves to \vars{start} when \vars{b}
is triggered. When in state \vars{start} for $d_1$ unit of time, the
observer moves to state \vars{watch} (this is the meaning of the
\code{wait} operator). The \code{select} operator is a
non-deterministic choice, with \code{unless} coding priorities. Hence,
in state \vars{watch}, the observer moves to \vars{ok} if an \vars{a}
occurs, unless a duration equals to $(d_2-d_1)$ elapses, in which case
it moves to the state \vars{error}. As a consequence, the pattern is
false whenever the probe \vars{(Present/state error)} is
reachable. Hence the formula associated to the pattern is $\phi_P
\eqdef \vars{[] - (Present/state error)}$.

To prove that an observer \vars{Obs} for the pattern \vars{P} is
correct, we need to prove that, for every system \vars{S}, the program
(\vars{S} \code{||} \vars{Obs}) satisfies the formula $\phi_{P}$ if
and only if $\sem{\vars{S}} \subseteq
\sem{\vars{P}}$. In~\cite{FRP11}, we have defined a mathematical
framework to formally prove these kind of properties, but this
framework relies on manual proofs and is not supported by any
tooling. Efforts are also under way to completely mechanize these
proofs using the Coq proof assistant~\cite{garnacho12}. Nonetheless,
formal proofs of correctness can be quite tedious. Therefore, to
detect possible problems with an observer early on (that is, before
spending a lot of efforts doing a formal proof of correctness) we also
rely on a ``visual'' verification method, that is akin to debugging
our observers.

In the next section, we show how to apply the visual verification
approach on our running example. One of the objective of our work is
to replace this visual verification step with a more formal
approach. This is done in Sect.~\ref{sec:autom-visu-verif}.

\section{Visual Verification of Observers}
\label{sec 4:patterns verification}

In the remainder of this section, we describe the visual verification
method using the particular case of the pattern \code{Present}
\vars{a} \code{after} \vars{b} \code{within} $[4, 5[$; we assume that
\vars{Obs} is the observer \vars{Present} defined in
Listing~\ref{lproc1}, that $d_1 = 4$ and that $d_2 = 5$. 

To prove that the observer \vars{Present} is correct, we need to
prove, for every system \vars{S}, the equivalence between two facts:
(1) the state \vars{(Present/state error)} is not reachable in the
program (\vars{S} \code{||} \vars{Present[a, b]}); and (2) the traces
of \vars{P} are valid for the property Pres, i.e. $\sem{\vars{S}}
\subseteq \sem{\mathrm{Pres}}$.

The first step is to get rid of the universal quantification on all
possible systems, \vars{S}, that is introduced by our definition of
correctness. The idea is to check the observer on a particular Fiacre
program---called \vars{Universal}---that can generate all possible
combinations of delays and events \vars{a}, \vars{b} and \vars{z}.  We
give an example of universal process in Listing~\ref{luniv}.
The process \vars{Universal} has only one state and three possible
transitions. Each transition changes the value of a shared integer
variable, \vars{x}. The first and second transitions of
\vars{Universal} can be fired without time constraints. In our
context, the probe \vars{a} will be triggered to the event ``setting
\vars{x} to 1'' and \vars{b} to ``setting \vars{x} to 2''. The third
transition reset the value of \vars{x} to 0 immediately and
corresponds to the internal event \vars{z}.

{
\begin{lstlisting}[float=tbph,label=luniv,captionpos=b,caption=\protect{Universal
    program in Fiacre},frame=single]
process Universal (&x : nat) is

   states s0

   from s0 select
	      x := 1; to s0      /* setting x to 1 */
	   [] x := 2; to s0      /* setting x to 2 */
	   unless
	      on (x <> 0); wait [0,0]; x := 0; to s0
	   end

	   
component Main is
   
   var x : nat := 0

   port a : sync is value (x = 1),  b : sync is value (x = 2)

   par Universal (&x) || Present [a, b] end
\end{lstlisting}}

We can now use our verification toolchain to generate the state graph
for the program (\vars{Universal} \code{||} \vars{Present}) using a
discrete time exploration construction. This can be obtained using the
flag \texttt{-F1} in {Tina} (it is possible to generate a state
graph with many different abstractions with Tina, including dense time
models).

The resulting graph is displayed in Fig.~\ref{fig:example}.  This
state graph has been generated and printed using the tool \emph{nd},
which is also part of the Tina toolset; nd is an editor and animator
for extended Time Petri Nets that can export nets and state graphs in
several, machine readable formats. This graph has only $26$ states and
can therefore be easily managed manually. The main factor commanding
the number of states is the value of the timing constraints used in
the pattern; in our observations, all the generated state graphs were
of manageable size.

The transitions in the state graph are also quite straightforward: we
find the visible and internal transitions as before, labeled with
\vars{a}, \vars{b}, \vars{z} and \vars{t}. For ease of reading, we
have also changed the labels of internal transitions in the observer
\vars{Present}. For instance, the transition from state 2 to 3
corresponds to the observer entering the state \vars{start}; likewise
for the transitions labeled with \vars{watch}, \vars{stop} and
\vars{error}. The states where the observer is in state \vars{error}
(the states that contradict the property $\phi_P \eqdef \vars{[] -
  (Present/state error)}$) are $\mathit{Errors} = \{20, 22, 23\}$.

\begin{figure}[ht]
\centering
\begin{tabular}[c]{l}
\includegraphics[width=0.9\textwidth]{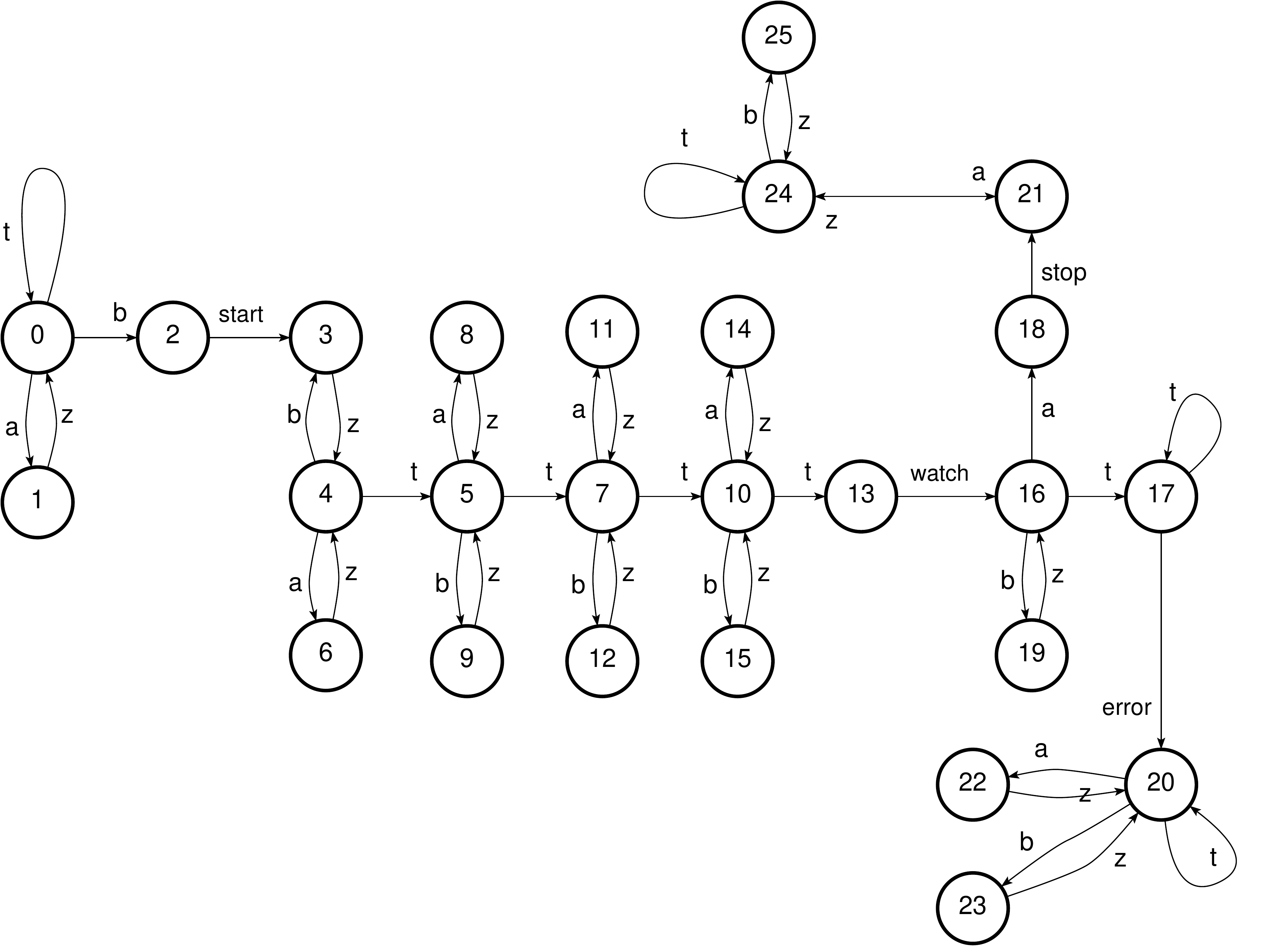}\\
\end{tabular}
\caption{State graph for (\vars{Universal} \code{||} \vars{Present})}
\label{fig:example}
\end{figure}

We can already debug the pattern \code{Present} \vars{a} \code{after}
\vars{b} \code{within} $[4,5[$ by visually inspecting the state
graph.

For \emph{soundness}, we need to check that, when the pattern is not
satisfied---for traces $\s$ that do not satisfy formula
$\mathrm{Pres}$---then the observer will detect a problem (observer
\vars{Present} eventually reaches a state in the set \emph{Errors}).

For \emph{innocuousness} we need to check that, from any state, it is
always possible to reach a state where event \vars{a} (respectively
\vars{b} and \vars{t}) can fire. Indeed, this means that the observer
cannot selectively remove the observation of a particular sequence of
external transitions or the passing of time. 

This graphical verification method has some drawbacks. As such, it
relies on a discrete time model and only works for fixed values of the
timing parameters (we have to fix the value of $d_1$ and
$d_2$). Nonetheless, it is usually enough to catch many errors in the
observer before we try to prove the observer correct more formally.

\section{Automating  the Visual Verification Method}
\label{sec:autom-visu-verif}

A problem with the previous approach is that it essentially relies on
an informal inspection (and on human interaction). We show how to
solve this problem by replacing the visual inspection of the state
graph by the verification of modal $\mu$-calculus formulas. (the Tina
toolset includes a model-checker for the $\mu$-calculus called
\emph{muse}.)  The general idea rests on the fact that we can
interpret the state graph as a finite state automaton and (some) sets
of traces as regular languages. This analogy is generally quite useful
when dealing with model-checking problems.  We start by defining some
useful notations.

\subsection{Label Expressions} 

Label expressions are boolean expressions denoting a set of
(transition) labels. For instance,
$A_\text{ext} = (\vars{a} \vee \vars{b})$ denotes the external
transitions, while the expression \ltl{-(\vars{a}\ltlor \vars{b}\ltlor
  \vars{t})} is only matched by the silent transition label. We will
also use the expression $\top$ to denote the conjunction of all
possible labels, e.g.
$\top = \text{\vars{(-b)}} \vee \text{\vars{b}}$. The model checker
\emph{muse} allows the definition of label expressions using the same
syntax.

\subsection{Regular (Path) Expressions} 

In the following, we consider regular expressions build from label
expressions. For example, the regular expression
$\vars{t} \cdot \text{\vars{(- t)}}\kleene $ denotes traces of
duration 1 with no events occurring at time $0$.
\begin{equation}
  \tick \ \eqdef\  \vars{t} \cdot
  \text{\vars{(-t)}}\kleene\label{eq:1}
\end{equation}
We remark that it is possible to define the set of (discrete) traces
where the FOTT formula {Pres} holds using the union of two regular
languages: (1) the traces where \vars{b} never occurs, $R_1 =
\text{\vars{(- b)}}\kleene$; and (2) the traces where there is an
\vars{a} four units of time after the first \vars{b}. In this
particular case, $R_2$ is a regular expression corresponding to the
property $(x = y \, \vars{b} \, z \, \vars{a}\, w) \wedge (\vars{b}
\notin y) \wedge (\Delta(z) \in [4,5[$\,)
\begin{eqnarray}
  Pres &\ \eqdef\ & \text{\vars{(- b)}}\kleene \ \vee\ R_2\\
  R_2 &\ \eqdef\ & \text{\vars{(-b)}}\kleene \cdot \vars{b} \cdot
  \text{\vars{(- t)}}\kleene \cdot \tick \cdot \tick \cdot \tick \cdot
  \tick \cdot \vars{a} \cdot \top\kleene\label{eqn:R2}
\end{eqnarray}
By construction, the regular language associated to $R_1 \vee R_2$ is
exactly the set of finite traces matching (the discrete semantics) of
Pres.  In the most general case, a regular expressions can always be
automatically generated from an existential FOTT formula when the time
constraints of delay expressions are fixed (the intervals $I$ in the
occurrences of ($\Delta(x) \in I$)\,).

The next step is to check that the observer agrees with every trace
conforming to $R_2$. For this we simply need to check that, starting
from the initial state of (\vars{Universal} \code{||} \vars{Present}),
it is not possible to reach a state in the set \emph{Errors} by
following a sequence of transitions labeled by a word in $R_2$. 

This is a simple instance of a language inclusion problem between
finite state automata. More precisely, if $\mathit{Present}$ is the
set of states visited when accepting the traces in $R_1 \vee R_2$, we
need to check that $\mathit{Errors}$ is included in the complement of
the set \emph{Present} (denoted $\overline{\mathit{Present}}$). In our
example of Fig.~\ref{fig:example}, we have that
$\overline{\mathit{Present}}= \{17, 20, 22, 23\}$, and therefore
$\mathit{Errors} \subseteq \overline{\mathit{Present}}$.

This automata-based approach has still some drawbacks. This is what
will motivate our use of a branching time logic in the next
section. In particular, this method is not enough to check the
soundness or the innocuousness of the observer. For innocuousness, we
need to check that every event may always eventually
happen. Concerning soundness, we need to prove that $\mathit{Errors}
\supseteq \overline{\mathit{Present}}$; which is false in our case.
The problem lies in the treatment of time divergence (and of
fairness), as can be seen from one of the counter-example produced
when we use our LTL model-checker to check the soundness property,
namely: \ltl{b.start.z.t.t.t.t.watch.t.t.$\cdots$} (ending with a
cycle of \texttt{t} transitions). This is an example where the error
transition is continuously enabled but never fired.

\subsection{Branching Time Specification} 

We show how to interpret regular expressions over traces using a modal
logic. In this case, the target logic is a modal $\mu$-calculus with
operators for forward and backward traversal of a state graph . (Many
temporal logics can be encoded in the $\mu$-calculus, including
CTL\kleene). In this context, the semantics of a formula $\psi$ over a
Kripke structure (a state graph) is the set of states where $\psi$
holds.
\[
\psi \ ::=\ \phi \wedge \psi \ \mid\ \neg \psi \ \mid \
\vars{<A>}\psi\ \mid\ \psi\vars{<A>} \ \mid\ X \ \mid\ (\min X\,|\,
\psi) \]

The basic modalities in the logic are \ltl{<$A$>$\psi$} and
\ltl{$\psi$<$A$>}, where $A$ is a label expression.  A state $s$ is in
$\text{\ltl{<$A$>}}\psi$ if and only if there is a (successor) state
$s'$ in $\psi$ and a transition from $s$ to $s'$ with a label in
A. Symmetrically, $s$ is in $\psi\text{\ltl{<$A$>}}$ if and only if
there is a (predecessor) state $s'$ in $\psi$ and a transition from
$s'$ to $s$ with a label in A. In the following, we will also use two
constants, \ltl{T}, the true formula (matching all the states), and
\ltl{`0}, that denotes the initial state of the model; and the least
fixpoint operator \ltl{min X | $\psi$(X)}.

For example, the formula \ltl{<\vars{a}>T} matches all the states that
are the source of an \vars{a}-transition, likewise
\ltl{Reach\_\vars{a}} \ \eqdef\ \ltl{min X\,|\,(<\vars{a}>T \ltlor
  <Z>X)} matches all the states that can lead to an
\vars{a}-transition using only internal transitions. As a consequence,
we can test innocuousness by checking that the formula
\ltl{(Reach\_\vars{a} \ltland Reach\_\vars{b} \ltland
  Reach\_\vars{t})} is true for all states.

The soundness proof rely on an encoding from regular path expressions
into modal formulas. We define two encodings: $\interp{R}$ that
matches the states encountered while firing a trace matching a regular
expression $R$; and $\interp[e]{R}$ that matches the state reached (at
the end) of a finite trace in $R$. These encodings rely on two derived
operators. (Again, we assume here that $A$ is a label expression.)
\[
\begin{array}[c]{l}
  \hfill {\psi \,\ltlo\, A  \ \eqdef\  \psi\ltl{<A>}}
  \hspace*{4em}
  {\psi \,\ltl{*}\, A \ \eqdef\ \ltl{min X | $\psi$\,\ltlor\,X<$A$>}}
  \hspace*{2em}\\[1em]

  \begin{array}[c]{lcl@{\quad}|@{\quad}lcl}

    \interp[e]{R \cdot A} & \eqdef & \interp[e]{R} \,\ltl{o}\, A &
    \interp{R \cdot A} & \eqdef & \interp{R} \vee \interp[e]{R
      \cdot A}\\

    \interp[e]{R \cdot A\kleene} & \ \eqdef\ & \interp[e]{R} \,\ltl{*}\, A &
    \interp{R \cdot A\kleene} & \ \eqdef\ & \interp{R} \vee \interp[e]{R \cdot A\kleene}\\
    
    
    \interp[e]{R \cdot \tick} & \ \eqdef\ &
    \ltl{(} \interp[e]{R} \ltl{\,o\,t)\,*\,(-t)} & 
    \interp{R \cdot \tick} & \ \eqdef\ &
    \interp{R} \vee \interp[e]{R \cdot \tick}\\

    \interp[e]{R_1 \vee R_2} & \ \eqdef\ &
    \interp[e]{R_1} \vee \interp[e]{R_2}  & 
    \interp{R_1 \vee R_2} & \ \eqdef\ &
    \interp{R_1} \vee \interp{R_2}\\

    \interp[e]{\epsilon} & \ \eqdef\ &
    \ltl{`0}
    & 
    \interp{\epsilon} & \ \eqdef\ &
    \ltl{`0}
  \end{array}\\
\end{array}
\]

\begin{lemma}
  Given a Kripke structure $K$, the states matching the formula
  $\interp[e]{R}$ (respectively $\interp{R}$\/) in $K$ are the states
  reachable from the initial state after firing (resp. all the states
  reachable while firing) a sequence of transitions matching $R$.
\end{lemma}
\begin{proof}[Sketch]
  By induction on the definition of $R$. For example, if we assume
  that $\psi$ correspond to the regular expression $R$, then $\psi
  \ltls A$ matches all the states reachable from states where $\psi$
  is true using (finite) sequences of transition with label in A;
  i.e. formula $\psi$ \ltl{*} $A$ corresponds to $R \cdot A\kleene$.
  Likewise, we use the interpretation of the empty expression,
  $\epsilon$, to prefix every formula with the constant \ltl{`0} (that
  will only match the initial state). This is necessary since
  $\mu$-calculus formulas are evaluated on all states whereas regular
  path expressions are evaluated from the initial state.\qed
\end{proof}

For example, we give the formula for $\interp[e]{R_2}$ below, where
$\psi\,\ltl{o}\,\ltl{Tick}$ stands for the expression
$(\psi\ltl{\,o\,t})\ltl{\,*\,(-t)}$:
\[
\interp[e]{R_2}  \ \eqdef\  \ltl{`0 * (-b) o b * (-t) o Tick o Tick o Tick o Tick o
  a * T}
\]
If $\psi_\mathit{Err}$ is a modal $\mu$-calculus formula that matches
the error condition of the observer, then we can check the correctness
and soundness of the observer \vars{Present} by proving that the
equivalence (EQ), below, is a tautology (that it is true on every
states of (\vars{Universal} \code{||} \vars{Present})).
\[
\tag{EQ}
\interp{\mathrm{Pres}} \ \Leftrightarrow\ - \psi_\mathit{Err}
\]
Again, we can interpret the ``error condition'' using the
$\mu$-calculus.  The definition of errors is a little bit more
involved than in the previous case. We say that a state is in error if
the transition \vars{error} is enabled (the formula \ltl{<error>T} is
true) or if the state can only be reached by firing the \vars{error}
transition (which corresponds to the formula
\ltl{(T<error>\,*\,T)\ltland\,(`0\,*\,(- error))}. Hence
$\psi_\mathit{Err}$ is the disjunction of these two properties:
\[
\psi_\mathit{Err} \quad \eqdef\quad \ltl{<error>T \ \ltlor\
  ((T<error>\,*\,T) \,\ltland\, - (`0\,*\,(-error)))}
\]
The formula (EQ) can be checked almost immediately (less than
\unit[1]{s} on a standard computer) for models of a few thousands
states using \emph{muse}. Listing~\ref{muse} gives a \emph{muse}
script file that can be used to test this equivalence relation.


\lstdefinelanguage{MMC}
{morekeywords={infix,op,output,set},
  sensitive=true,
  literate=%
  {=}{{\(=\,\)}}1%
  {/\\}{{\(\wedge\,\)}}1%
  {\\/}{{\(\vee\,\)}}1%
  {<}{{\(\langle\)}}1%
  {>}{{\(\rangle\)}}1%
  {<=>}{{\(\,\Leftrightarrow\,\)}}3%
}

\lstset{language=MMC,
        inputencoding=latin1,
        keywordstyle=\bf\sffamily,
        keywordstyle={[2]\itshape},
        keywordstyle=\color{magenta},
        morecomment=[l]{\#},
        columns=flexible,
        basicstyle=\sffamily,
        commentstyle=\color{codegreen},
        numberstyle=\scriptsize\sffamily\itshape,
        aboveskip=-2pt,
        belowskip=-4pt,
        mathescape=true, 
        texcl=true, 
        escapechar=@,
        numberblanklines=false
      }


\begin{lstlisting}[float=tbph,label=muse,frame=single,captionpos=b,caption=\protect{Script
  file for \emph{muse} to check that $\interp{\mathrm{Pres}}
  \Leftrightarrow - \psi_\mathit{Err}$ is a tautology}]
# Results are displayed as set of states. Use "output card" to see the cardinality
output set;

# definition of derived operators

infix X * L  = min Y | X \/ Y<L>;	infix X o L  = X<L>;
op TICK X  = min Y | X<t> \/ Y<-t>;	op NEVER L = (`0) * (-L);
op EXT = a \/ b \/ t; # labels of the external transitions
op REACH L = min X | (<L>T) \/ <-EXT>X;	

# INNOCUOUSNESS

op Innocuous = (REACH a) /\ (REACH b) /\ (REACH t);

# SOUNDNESS

op A0 = (NEVER b) o b;		op S0 = (NEVER b) \/ A0;
op A1 = A * (-t);   		op S1 = S0 \/ A1;
op A2 = TICK(A1);		op S2 = S1 \/ A2;
op A3 = TICK(A2);		op S3 = S2 \/ A3;
op A4 = TICK(A3);		op S4 = S3 \/ A4;
op A5 = TICK(A4);		op S5 = S4 \/ A5;
op A6 = A5 o a;			op S6 = S5 \/ A6;
op A7 = A6 * T;			op S7 = S6 \/ A7;

op R1 = NEVER b;		op R2 = S7		
op Pres = R1 \/ R2;
op ERRORS = <error>T \/ (((T<error>) * T) /\ - ((`0) * (-error)));

Pres <=> (- ERRORS); # this is a tautology if all the states are listed

\end{lstlisting}

\section{Related Work and Conclusion}
\label{sec 6:related-work}

Few works consider the verification of model-checking tools. Indeed,
most of the existing approaches concentrate on the verification of the
model-checking algorithms, rather than on the verification of the
tools themselves. For example, Smaus et al.~\cite{smaus09} provide a
formal proof of an algorithm for generating Büchi automata from a LTL
formula using the Isabelle interactive theorem prover. This algorithm
is at the heart of many LTL model-checker based on an
automata-theoretic approach. The problem of verifying verification
tools also appears in conjunction with certification issues. In
particular, many certification norms, such as the DO-178B, requires
that any tool used for the development of a critical equipment be
qualified at the same level of criticality than the equipment. (Of
course, certification does not necessarily mean formal proof!) In this
context, we can cite the work done on the certification of the SCADE
compiler~\cite{scade}, a tool-suite based on the synchronous language
Lustre that integrates a model-checking engine. Nonetheless, only the
code-generation part of the compiler is certified and not the
verification part.

Concerning observer-based model-checking, most of the works rely on an
automatic way to synthesize observers from a formal definition of the
properties. For instance, Aceto et al.~\cite{MCRTTA} propose a method
to verify properties based on the use of test automata. In this
framework, verification is limited to safety and bounded liveness
properties since the authors focus on properties that can be reduced
to reachability checking. In the context of Time Petri Net, Toussaint
et al.~\cite{TCVMBTPN} also propose a verification technique based on
``timed observers'', but they only consider four specific kinds of
time constraints. None of these works consider the complexity or the
correctness of the verification problem. Another related work
is~\cite{PTPS}, where the authors define observers based on Timed
Automata for each pattern.  Our approach is quite orthogonal to the
``synthesis approach''. Indeed we seek, for each property, to come up
with the best possible observer in practice. To this end, using our
toolchain, we compare the complexity of different implementations on a
fixed set of representative examples and for a specific set of
properties and kept the best candidates. The need to check multiple
implementations for the same patterns has motivated the need to
develop a lightweight verification method for checking their
correctness.  

Compared to these works, we make several contributions. We define a
complete verification framework for checking observers with hard
realtime constraints. This framework has been tested on a set of
observers derived from high-level timed specification patterns. This
work is also our first public application of the probe technology,
that was added to Fiacre only recently. To the best of our knowledge,
the notion of \emph{probes} is totally new in the context of formal
specification language. Paun and Chechik propose a somewhat similar
mechanism in~\cite{epp99,eeltp99}---in an untimed setting---where they
define new categories of events. However our approach is more general,
as we define probes for a richer set of events, such as variables
changing state. We believe that this (language-level) notion of probes
is interesting in its own right and could be adopted by other formal
specification languages. Finally, we propose a formal approach that
can be used to gain confidence on the implementation of our
model-checking tools and that replaces traditional ``visual
verification methods'' that are prone to human errors. This result
also prove the usefulness of having access to a complete toolbox that
provides different kind of tools: editors, model-checkers for
different kind of logics, \dots

 

{\footnotesize
\begin{thebibliography}{10}


\bibitem{FRP11} N.~Abid, S.~Dal~Zilio, and D.~Le~Botlan.  \newblock {A
    formal framework to specify and verify real-time properties on
    critical systems}.  \newblock International Journal of Critical
  Computer-Based Systems (IJCCBS) 5(1/2):4-30, 2014.


\bibitem{MCRTTA}
L.~Aceto, A.~Burgue{\~n}o, and K.~G. Larsen.
\newblock Model checking via reachability testing for timed automata.
\newblock In {\em Proc. of TACAS}, vol. 1384 of {\em LNCS}. Springer, 1998.


\bibitem{Fiacre07}
B.~Berthomieu, J.-P.~Bodeveix, and M.~Fillali and G.~Hubert and F.~Lang and F.~Peres and R.~Saad and S.~Jan and F.~Vernadat.
\newblock The {S}yntax and {S}emantics of {F}iacre -- {V}ersion 3.0.
\newblock \url{http://www.laas.fr/fiacre/}, 2012.

\bibitem{tina}
B.~Berthomieu, P.-O. Ribet, and F.~Vernadat.
\newblock The tool Tina -- construction of abstract state spaces for {Petri}
  nets and time {Petri} nets.
\newblock {\em International Journal of Production Research}, 42:14, 2004.

\bibitem{filfmvte2008}
B.~Berthomieu, J.-P. Bodeveix, P.~Farail, M.~Filali, H.~Garavel, P.~Gaufillet,
  F.~Lang, and F.~Vernadat.
\newblock {Fiacre: an Intermediate Language for Model Verification in the
  Topcased Environment}.
\newblock In {\em Proc. of {ERTS}}, 2008.

\bibitem{epp99}
M.~Chechik and D.O.~Paun.
\newblock Events in Property Patterns.
\newblock In {\em Theoretical and Practical Aspects of SPIN Model Checking}, vol.~3, 1999.

\bibitem{ksu} M.~B. Dwyer, L.~Dillon.
\newblock Online Repository of Specification Patterns. 
\newblock At \url{http://patterns.projects.cis.ksu.edu/}

\bibitem{garnacho12} M.~Garnacho, J.-P.~Bodeveix and
  M.~Filali. \newblock A Mechanized Semantic Framework for Real-Time
  Systems. \newblock In {\em Proc. of FORMATS}, LNCS vol.~8053, 2013.

\bibitem{PTPS}
V.~Gruhn and R.~Laue.
\newblock Patterns for timed property specifications.
\newblock {\em Electr. Notes Theor. Comput. Sci.}, 153(2):117--133, 2006.


\bibitem{janowska2011towards} A Janowska, W Penczek,
  A. P{\'o}{\l}rola, A Zbrzezny. \newblock Towards discrete-time
  verification of time Petri nets with dense-time semantics. \newblock
  In {\em Proc. of the Int. Workshop on Concurrency, Specification and
    Programming}, 2011.

\bibitem{SRPMTL}
R.~Koymans.
\newblock Specifying realtime properties with metric temporal logic.
\newblock {\em Realtime Syst.}, 2:255--299,  1990.

\bibitem{merlin}
P.~M. Merlin.
\newblock {\em A study of the recoverability of computing systems.}
\newblock PhD thesis, 1974.

\bibitem{ow05}
J.~Ouaknine and J.~Worrell.
\newblock On the decidability and complexity of metric temporal logic over
  finite words.
\newblock In {\em Logical Methods in Computer Science}, vol.~3, 2007.

\bibitem{eeltp99}
D.O.~ Paun and M.~Chechik.
\newblock Events in Events in Linear-Time Properties.
\newblock In {\em CoRR journal}, vol.~cs.SE/9906031, 1999.

\bibitem{scade} Esterel Technologies. \newblock SCADE Tool
  Suite.\newblock
  \url{http://www.esterel-technologies.com/products/scade-suite}

\bibitem{smaus09} A. Schimpf, S. Merz and J.-G. Smaus.  \newblock
  Construction of Büchi Automata for LTL Model Checking Verified in
  Isabelle/HOL. \newblock In {\em Proc. of TPHOLs}, LNCS vol.~5674,
  2009.

\bibitem{TCVMBTPN}
J.~Toussaint, F.~Simonot-Lion, and J.-P. Thomesse.
\newblock Time constraints verification methods based on time {Petri} nets.
\newblock In {\em Proc. of FTDCS}. IEEE, 1997.
\end{thebibliography}}
\end{document}